\documentclass[letterpaper, 10 pt, conference]{ieeeconf}  

\IEEEoverridecommandlockouts                              
\usepackage{graphicx,amssymb,amstext}
\usepackage{amsmath}
\usepackage{hyperref}
\usepackage{epsfig}
\usepackage{verbatim}
\usepackage{cite}

\newtheorem{theorem}{Theorem}
\newtheorem{definition}{Definition}
\newtheorem{remark}{Remark}

\newcommand{\qed}{\nobreak \ifvmode \relax \else
      \ifdim\lastskip<1.5em \hskip-\lastskip
      \hskip1.5em plus0em minus0.5em \fi \nobreak
      \vrule height0.3em width0.5em depth0.25em\fi}

\title{\LARGE \bf
Coherent-Classical Estimation versus Purely-Classical Estimation for Linear Quantum Systems
}

\author{Shibdas Roy*, Ian R. Petersen and Elanor H. Huntington%
\thanks{This work was supported by the Australian Research Council (ARC).}%
\thanks{The authors are with the School of Engineering and Information Technology, University of New South Wales at the Australian Defence Force Academy, Canberra, ACT 2600, Australia.}%
\thanks{*\tt\small shibdas.roy at student.adfa.edu.au}%
}

\begin{document}

\maketitle
\thispagestyle{empty}
\pagestyle{empty}

\begin{abstract}

We consider a coherent-classical estimation scheme for a class of linear quantum systems. It comprises an estimator that is a mixed quantum-classical system without involving coherent feedback. The estimator yields a classical estimate of a variable for the quantum plant. We demonstrate that for a passive plant that can be characterized by annihilation operators only, such coherent-classical estimation provides no improvement over purely-classical estimation. An example is also given which shows that if the plant is not assumed to be an annihilation operator only quantum system, it is possible to get better estimates with such coherent-classical estimation compared with purely-classical estimation.

\end{abstract}

\section{Introduction}

\bstctlcite{BSTcontrol}

Estimation and control problems for quantum systems have been of significant interest in recent years \cite{WM1,YK1,YK2,NY,JNP,NJP,GGY,MP1,MP2,YNJP,GJ,GJN,IRP1}. An important class of quantum systems are linear quantum systems \cite{WM1,NY,JNP,NJP,GGY,GJN,WM2,GZ,WD,NJD,HM,SSM,IRP3}, that describe quantum optical devices such as optical cavities \cite{WM2,BR}, linear quantum amplifiers \cite{GZ}, and finite bandwidth squeezers \cite{GZ}. Much recent work has considered coherent feedback control for linear quantum systems, where the feedback controller itself is also a quantum system \cite{JNP,NJP,MP1,MP2,HM,WM3,SL,GW,IRP3}. A related coherent-classical estimation problem has been considered by one of the authors in Ref. \cite{IRP2}, where the estimator consists of a classical part, which produces the desired final estimate and a quantum part, which may also involve coherent feedback. Note that this is different from the problem considered in Ref. \cite{MJ} which involves constructing a quantum observer. A quantum observer is a purely quantum system, that produces a quantum estimate of a variable for the quantum plant. On the other hand, a coherent-classical estimator is a mixed quantum-classical system, that produces a classical estimate of a variable for the quantum plant.

In this paper, we consider a special case of the coherent-classical estimation problem, i.e. one that does not involve coherent feedback. We first formulate the optimal coherent-classical estimation problem without any assumptions. We then illustrate that if the quantum plant is assumed to be a physically realizable annihilation operator only quantum system, there is no improvement in the accuracy of the estimate with a coherent-classical estimator over that with a classical-only estimator. We present an example where this is the case to illustrate this result. However, if the quantum plant is physically realizable but not assumed to be characterized by annihilation operators only, it is possible to get more precise estimate with a coherent-classical estimator when compared with a purely-classical estimator as demonstrated by an example which is presented.

\section{Linear Quantum Systems and Physical Realizability}
The class of linear quantum systems we consider here are described by the quantum stochastic differential equations (QSDEs) \cite{JNP,GJN,IRP1,IRP2,SP}:
\begin{equation}\label{eq:lqs_1}
\begin{split}
\left[\begin{array}{c}
da(t)\\
da(t)^{\#}
\end{array}\right] &= F \left[\begin{array}{c}
a(t)\\
a(t)^{\#}
\end{array}\right] dt + G \left[\begin{array}{c}
d\mathcal{A}(t)\\
d\mathcal{A}(t)^{\#}
\end{array}\right];\\
\left[\begin{array}{c}
d\mathcal{A}^{out}(t)\\
d\mathcal{A}^{out}(t)^{\#}
\end{array}\right] &= H \left[\begin{array}{c}
a(t)\\
a(t)^{\#}
\end{array}\right] dt + K \left[\begin{array}{c}
d\mathcal{A}(t)\\
d\mathcal{A}(t)^{\#}
\end{array}\right],
\end{split}
\end{equation}
where
\begin{equation}\label{eq:lqs_2}
\begin{split}
F &= \Delta(F_1,F_2), \qquad G = \Delta(G_1,G_2),\\
H &= \Delta(H_1,H_2), \qquad K = \Delta(K_1,K_2).
\end{split}
\end{equation}

Here, $a(t) = [a_1(t) \hdots a_n(t)]^T$ is a vector of annihilation operators. The adjoint of the operator $a_i$ is called a creation operator, denoted by $a_i^{*}$. Also, the notation $\Delta(F_1,F_2)$ denotes the matrix $\left[\begin{array}{cc} F_1 & F_2\\
F_2^{\#} & F_1^{\#}
\end{array}\right]$. Here, $F_1$, $F_2 \in \mathbb{C}^{n \times n}$, $G_1$, $G_2 \in \mathbb{C}^{n \times m}$, $H_1$, $H_2 \in \mathbb{C}^{m \times n}$, and $K_1$, $K_2 \in \mathbb{C}^{m \times m}$. Moreover, $^{\#}$ denotes the adjoint of a vector of operators or the complex conjugate of a complex matrix. Furthermore, $^\dagger$ denotes the adjoint transpose of a vector of operators or the complex conjugate transpose of a complex matrix. In addition, the vector $\mathcal{A} = [\mathcal{A}_1 \hdots \mathcal{A}_m]^T$ represents a collection of external independent quantum field operators and the vector $\mathcal{A}^{out}$ represents the corresponding vector of output field operators.

\begin{definition}
(See \cite{JNP,IRP1,IRP2,SP}) A complex linear quantum system of the form (\ref{eq:lqs_1}), (\ref{eq:lqs_2}) is said to be physically realizable if there exists a complex commutation matrix $\Theta = \Theta^\dagger$, a complex Hamiltonian matrix $M = M^\dagger$, and a coupling matrix $N$ such that
\begin{equation}\label{eq:theta}
\Theta = TJT^\dagger,
\end{equation}
where $J = \left[\begin{array}{cc}
I & 0\\
0 & -I
\end{array}\right]$, $T = \Delta(T_1,T_2)$ is non-singular, $M$ and $N$ are of the form
\begin{equation}
M = \Delta(M_1,M_2), \qquad N = \Delta(N_1,N_2)
\end{equation}
and
\begin{equation}
\begin{split}
F &= -\iota\Theta M - \frac{1}{2}\Theta N^\dagger JN,\\
G &= -\Theta N^\dagger J,\\
H &= N,\\
K &= I.
\end{split}
\end{equation}
\end{definition}

Here, the commutation matrix $\Theta$ satisfies the following commutation relation:
\begin{equation}\label{eq:comm_rel1}
\begin{split}
&\left[\left[\begin{array}{c}
a\\
a^{\#}
\end{array}\right], \left[\begin{array}{c}
a\\
a^{\#}
\end{array}\right]^\dagger\right]\\
&= \left[\begin{array}{c}
a\\
a^{\#}
\end{array}\right]\left[\begin{array}{c}
a\\
a^{\#}
\end{array}\right]^\dagger - \left(\left[\begin{array}{c}
a\\
a^{\#}
\end{array}\right]^\# \left[\begin{array}{c}
a\\
a^{\#}
\end{array}\right]^T\right)^T\\
&= \Theta .
\end{split}
\end{equation}

\begin{theorem}\label{thm:phys_rlz}
(See \cite{SP,IRP3}) The linear quantum system (\ref{eq:lqs_1}), (\ref{eq:lqs_2}) is physically realizable if and only if there exists a complex matrix $\Theta = \Theta^\dagger$ such that $\Theta$ is of the form in (\ref{eq:theta}), and
\begin{equation}
\begin{split}
F\Theta + \Theta F^\dagger + GJG^\dagger &= 0,\\
G &= -\Theta H^\dagger J,\\
K &= I.
\end{split}
\end{equation}
\end{theorem}

If the system (\ref{eq:lqs_1}) is physically realizable, then the matrices $M$ and $N$ define a complex open harmonic oscillator with coupling operator
\[ \mathbf{L} = \left[\begin{array}{cc}
N_1 & N_2
\end{array}\right] \left[\begin{array}{c}
a\\
a^{\#}
\end{array}\right],\]
and a Hamiltonian operator
\[ \mathbf{H} = \frac{1}{2} \left[\begin{array}{cc}
a^\dagger & a^T
\end{array}\right] M \left[\begin{array}{c}
a\\
a^{\#}
\end{array}\right].\]

\subsection{Annihilation Operator Only Linear Quantum Systems}

Annihilation operator only linear quantum systems are a special case of the above class of linear quantum systems, where the QSDEs (\ref{eq:lqs_1}) can be described purely in terms of the vector of annihilation operators \cite{MP1,MP2}. In this case, we consider Hamiltonian operators of the form $\mathbf{H} = a^\dagger M a$ and coupling vectors of the form $\mathbf{L} = N a$, where $M$ is a Hermitian matrix and $N$ is a complex matrix. The commutation relation (\ref{eq:comm_rel1}), in this case, takes the form:
\begin{equation}\label{eq:comm_rel2}
\left[a,a^\dagger \right] = \Theta,
\end{equation}
where $\Theta > 0$. Also, the corresponding QSDEs are given by:
\begin{equation}\label{eq:ann_lqs_1}
\begin{split}
da &= Fadt + Gd\mathcal{A};\\
d\mathcal{A}^{out} &= Hadt + Kd\mathcal{A}.
\end{split}
\end{equation}

\begin{definition}
(See \cite{MP1,IRP3}) A linear quantum system of the form (\ref{eq:ann_lqs_1}) is physically realizable if there exist complex matrices $\Theta>0$, $M=M^\dagger$, $N$, such that the following is satisfied: 
\begin{equation}
\begin{split}
F &= \Theta \left(-\iota M - \frac{1}{2} N^\dagger N\right),\\
G &= -\Theta N^\dagger ,\\
H &= N,\\
K &= I.
\end{split}
\end{equation}
\end{definition}

\begin{theorem}\label{thm:ann_phys_rlz}
(See \cite{MP1,IRP3}) An annihilation operator only linear quantum system of the form (\ref{eq:ann_lqs_1}) is physically realizable if and only if there exists a complex matrix $\Theta>0$ such that
\begin{equation}
\begin{split}
F\Theta + \Theta F^\dagger + GG^\dagger &= 0,\\
G &= -\Theta H^\dagger ,\\
K &= I.
\end{split}
\end{equation}
\end{theorem}

\subsection{Linear Quantum System from Quantum Optics}

An example of a linear quantum system from quantum optics is a dynamic squeezer, that is an optical cavity with a non-linear optical element inside as shown in Fig. \ref{fig:sqz_scm}. Upon suitable linearizations and approximations, such an optical squeezer can be described by the following quantum stochastic differential equations \cite{IRP2}:
\begin{equation}
\begin{split}
da &= -\frac{\gamma}{2}adt - \chi a^{*} dt - \sqrt{\kappa_1}d\mathcal{A}_1 - \sqrt{\kappa_2}d\mathcal{A}_2;\\
d\mathcal{A}_1^{out} &= \sqrt{\kappa_1}adt + d\mathcal{A}_1;\\
d\mathcal{A}_2^{out} &= \sqrt{\kappa_2}adt + d\mathcal{A}_2,
\end{split}
\end{equation}
where $\kappa_1,\kappa_2>0$, $\chi \in \mathbb{C}$, and $a$ is a single annihilation operator of the cavity mode \cite{BR,GZ}. This leads to a linear quantum system of the form (\ref{eq:lqs_1}) as follows:
\begin{equation}
\begin{split}
\left[\begin{array}{c}
da(t)\\
da(t)^{*}
\end{array}\right] &= \left[\begin{array}{cc}
-\frac{\gamma}{2} & -\chi\\
-\chi^{*} & -\frac{\gamma}{2}
\end{array}\right] \left[\begin{array}{c}
a(t)\\
a(t)^{*}
\end{array}\right] dt\\
&\quad - \sqrt{\kappa_1} \left[\begin{array}{c}
d\mathcal{A}_1(t)\\
d\mathcal{A}_1(t)^{\#}
\end{array}\right]\\
&\quad - \sqrt{\kappa_2} \left[\begin{array}{c}
d\mathcal{A}_2(t)\\
d\mathcal{A}_2(t)^{\#}
\end{array}\right];\\
\left[\begin{array}{c}
d\mathcal{A}_1^{out}(t)\\
d\mathcal{A}_1^{out}(t)^{\#}
\end{array}\right] &= \sqrt{\kappa_1} \left[\begin{array}{c}
a(t)\\
a(t)^{*}
\end{array}\right] dt + \left[\begin{array}{c}
d\mathcal{A}_1(t)\\
d\mathcal{A}_1(t)^{\#}
\end{array}\right];\\
\left[\begin{array}{c}
d\mathcal{A}_2^{out}(t)\\
d\mathcal{A}_2^{out}(t)^{\#}
\end{array}\right] &= \sqrt{\kappa_2} \left[\begin{array}{c}
a(t)\\
a(t)^{*}
\end{array}\right] dt + \left[\begin{array}{c}
d\mathcal{A}_2(t)\\
d\mathcal{A}_2(t)^{\#}
\end{array}\right].
\end{split}
\end{equation}

The above quantum system requires $\gamma = \kappa_1 + \kappa_2$ in order for the system to be physically realizable.

\begin{figure}[!b]
\centering
\includegraphics[width=0.3\textwidth]{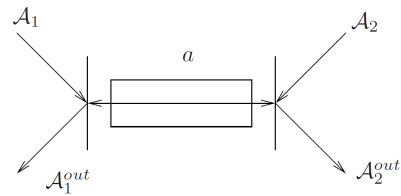}
\caption{Schematic diagram of a dynamic optical squeezer.}
\label{fig:sqz_scm}
\end{figure}

Also, the above quantum optical system can be described purely in terms of the annihilation operator only if $\chi = 0$, i.e. there is no squeezing, in which case it reduces to a passive optical cavity. This leads to a linear quantum system of the form (\ref{eq:ann_lqs_1}) as follows:
\begin{equation}
\begin{split}
da &= -\frac{\gamma}{2}adt - \sqrt{\kappa_1} d\mathcal{A}_1 - \sqrt{\kappa_2} d\mathcal{A}_2;\\
d\mathcal{A}_1^{out} &= \sqrt{\kappa_1} adt + d\mathcal{A}_1;\\
d\mathcal{A}_2^{out} &= \sqrt{\kappa_2} adt + d\mathcal{A}_2,
\end{split}
\end{equation}
where again the system is physically realizable when we have $\gamma = \kappa_1 + \kappa_2$.

\section{Purely-Classical Estimation}
The schematic diagram of a purely-classical estimation scheme is provided in Fig. \ref{fig:cls_scm}. We consider a quantum plant, which is a quantum system of the form (\ref{eq:lqs_1}), (\ref{eq:lqs_2}), defined as follows:
\begin{equation}\label{eq:plant}
\begin{split}
\left[\begin{array}{c}
da(t)\\
da(t)^{\#}
\end{array}\right] &= F \left[\begin{array}{c}
a(t)\\
a(t)^{\#}
\end{array}\right] dt + G \left[\begin{array}{c}
d\mathcal{A}(t)\\
d\mathcal{A}(t)^{\#} 
\end{array}\right];\\
\left[\begin{array}{c}
d\mathcal{Y}(t)\\
d\mathcal{Y}(t)^{\#}
\end{array}\right] &= H \left[\begin{array}{c}
a(t)\\
a(t)^{\#}
\end{array}\right] dt + K \left[\begin{array}{c}
d\mathcal{A}(t)\\
d\mathcal{A}(t)^{\#}
\end{array}\right];\\
z &= C\left[\begin{array}{c}
a(t)\\
a(t)^{\#} 
\end{array}\right].
\end{split}
\end{equation}

\begin{figure}[!b]
\centering
\includegraphics[width=0.5\textwidth]{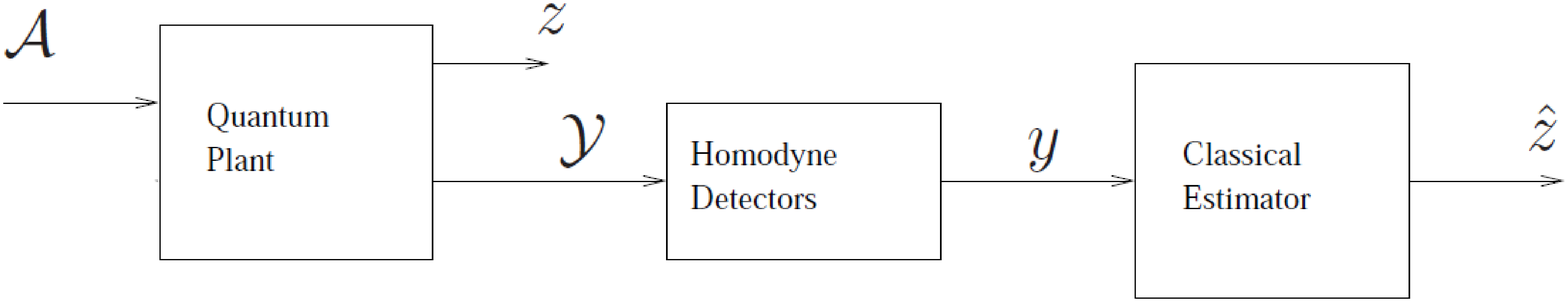}
\caption{Schematic diagram of purely-classical estimation.}
\label{fig:cls_scm}
\end{figure}

Here, $z$ denotes a scalar operator on the underlying Hilbert space and represents the quantity to be estimated. Also, $\mathcal{Y}(t)$ represents the vector of output fields of the plant, and $\mathcal{A}(t)$ represents a vector of quantum noises acting on the plant.

In the case of a purely-classical estimator, a quadrature of each component of the vector $\mathcal{Y}(t)$ is measured using homodyne detection to produce a corresponding classical signal $y_i$:
\begin{equation}\label{eq:class_hd}
\begin{split}
dy_1 &= \cos(\theta_1)d\mathcal{Y}_1 + \sin(\theta_1)d\mathcal{Y}_1^{*};\\
&\vdots\\
dy_m &= \cos(\theta_m)d\mathcal{Y}_m + \sin(\theta_m)d\mathcal{Y}_m^{*}.
\end{split}
\end{equation}

Here, the angles $\theta_1,\hdots,\theta_m$ determine the quadrature measured by each homodyne detector. The vector of classical signals $y = [y_1 \hdots y_m]^T$ is then used as the input to a classical estimator defined as follows:
\begin{equation}\label{eq:class_estimator}
\begin{split}
dx_e &= F_ex_edt + G_edy;\\
\hat{z} &= H_ex_e.
\end{split}
\end{equation}

Here $\hat{z}$ is a scalar classical estimate of the quantity $z$. The estimation error corresponding to this estimate is
\begin{equation}\label{eq:est_err}
e = z - \hat{z}.
\end{equation}

Then, the optimal classical estimator is defined as the system (\ref{eq:class_estimator}) that minimizes the quantity
\begin{equation}\label{eq:est_cost}
\bar{J}_c = \lim_{t \to \infty} \left\langle e^{*}(t)e(t) \right\rangle ,
\end{equation}
which is the mean-square error of the estimate. Here, $\langle \cdot \rangle$ denotes the quantum expectation over the joint quantum-classical system defined by (\ref{eq:plant}), (\ref{eq:class_hd}), (\ref{eq:class_estimator}).

The optimal classical estimator is given by the standard (complex) Kalman filter defined for the system (\ref{eq:plant}), (\ref{eq:class_hd}). This optimal classical estimator is obtained from the solution to the algebraic Riccati equation:
\begin{equation}\label{eq:riccati}
\begin{split}
F_a\bar{P}_e &+ \bar{P}_eF_a^\dagger + G_aG_a^\dagger - (G_aK_a^\dagger + \bar{P}_eH_a^\dagger)L^\dagger \\
&\times(LK_aK_a^\dagger L^\dagger)^{-1}L(G_aK_a^\dagger + \bar{P}_eH_a^\dagger)^\dagger = 0,
\end{split}
\end{equation}
where
\begin{equation}
\begin{split}
F_a &= F, \qquad G_a = G,\\
H_a &= H, \qquad K_a = K,\\
L &= \left[\begin{array}{cc}
L_1 & L_2
\end{array}\right],\\
L_1 &= \left[\begin{array}{cccc}
\cos(\theta_1) & 0 & \hdots & 0\\
0 & \cos(\theta_2) & \hdots & 0\\
 & & \ddots & \\
 & & & \cos(\theta_m)
\end{array}\right],\\
L_2 &= \left[\begin{array}{cccc}
\sin(\theta_1) & 0 & \hdots & 0\\
0 & \sin(\theta_2) & \hdots & 0\\
 & & \ddots & \\
 & & & \sin(\theta_m)
\end{array}\right].
\end{split}
\end{equation}

Here we have assumed that the quantum noise $\mathcal{A}$ is purely canonical, i.e. $d\mathcal{A}d\mathcal{A}^\dagger = Idt$ and hence $K=I$.

Eq. (\ref{eq:riccati}) thus becomes:

\begin{equation}\label{eq:class_riccati}
\begin{split}
F\bar{P}_e &+ \bar{P}_eF^\dagger + GG^\dagger - (G + \bar{P}_eH^\dagger)L^\dagger L(G + \bar{P}_eH^\dagger)^\dagger =0,
\end{split}
\end{equation}

The value of the cost (\ref{eq:est_cost}) is given by
\begin{equation}\label{eq:class_cost}
\bar{J}_c = C\bar{P}_eC^\dagger .
\end{equation}

\section{Coherent-Classical Estimation}
The schematic diagram of the coherent-classical estimation scheme under consideration is provided in Fig. \ref{fig:coh_cls_scm}. In this case, the plant output $\mathcal{Y}(t)$ does not directly drive a bank of homodyne detectors as in (\ref{eq:class_hd}). Rather, this output is fed into another quantum system called a coherent controller, which is defined as follows:
\begin{equation}\label{eq:coh_controller}
\begin{split}
\left[\begin{array}{c}
da_c(t)\\
da_c(t)^{\#}
\end{array}\right] &= F_c \left[\begin{array}{c}
a_c(t)\\
a_c(t)^{\#}
\end{array}\right] dt + G_c \left[\begin{array}{c}
d\mathcal{Y}(t)\\
d\mathcal{Y}(t)^{\#} 
\end{array}\right];\\
\left[\begin{array}{c}
d\tilde{\mathcal{Y}}(t)\\
d\tilde{\mathcal{Y}}(t)^{\#}
\end{array}\right] &= H_c \left[\begin{array}{c}
a_c(t)\\
a_c(t)^{\#}
\end{array}\right] dt + K_c \left[\begin{array}{c}
d\mathcal{Y}(t)\\
d\mathcal{Y}(t)^{\#} 
\end{array}\right].
\end{split}
\end{equation}

\begin{figure}[!b]
\centering
\includegraphics[width=0.5\textwidth]{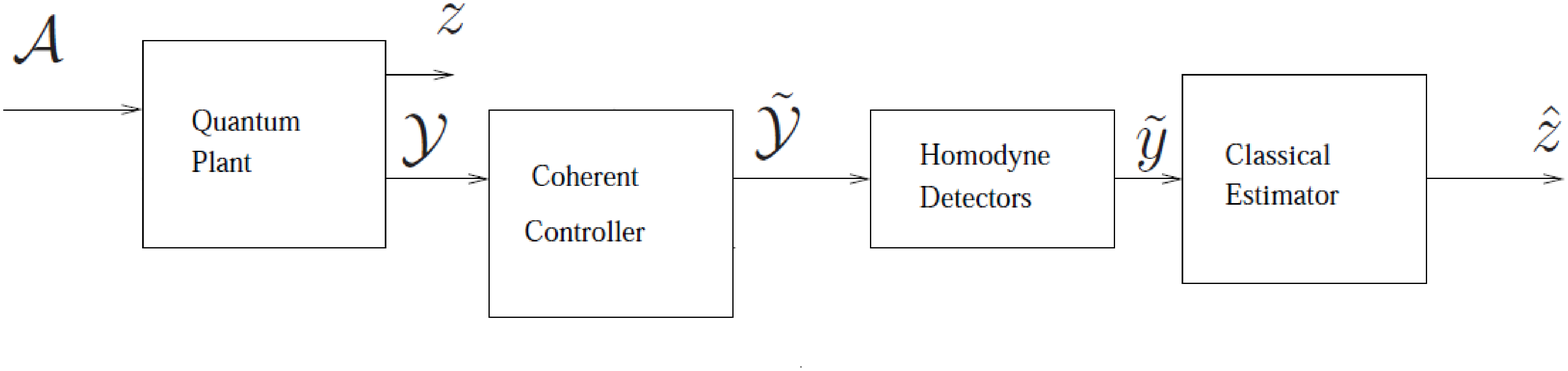}
\caption{Schematic diagram of coherent-classical estimation.}
\label{fig:coh_cls_scm}
\end{figure}

A quadrature of each component of the vector $\tilde{\mathcal{Y}}(t)$ is measured using homodyne detection to produce a corresponding classical signal $\tilde{y}_i$:
\begin{equation}\label{eq:coh_class_hd}
\begin{split}
d\tilde{y}_1 &= \cos(\tilde{\theta}_1)d\tilde{\mathcal{Y}}_1 + \sin(\tilde{\theta}_1)d\tilde{\mathcal{Y}}_1^{*};\\
&\vdots\\
d\tilde{y}_{\tilde{m}} &= \cos(\tilde{\theta}_{\tilde{m}})d\tilde{\mathcal{Y}}_{\tilde{m}} + \sin(\tilde{\theta}_{\tilde{m}})d\tilde{\mathcal{Y}}_{\tilde{m}}^{*}.
\end{split}
\end{equation}

Here, the angles $\tilde{\theta}_1,\hdots,\tilde{\theta}_{\tilde{m}}$ determine the quadrature measured by each homodyne detector. The vector of classical signals $\tilde{y} = [\tilde{y}_1 \hdots \tilde{y}_{\tilde{m}}]^T$ is then used as the input to a classical estimator defined as follows:
\begin{equation}\label{eq:coh_class_estimator}
\begin{split}
d\tilde{x}_e &= \tilde{F}_e\tilde{x}_edt + \tilde{G}_ed\tilde{y};\\
\hat{z} &= \tilde{H}_e\tilde{x}_e.
\end{split}
\end{equation}

Here $\hat{z}$ is a scalar classical estimate of the quantity $z$. Corresponding to this estimate is the estimation error (\ref{eq:est_err}). Then, the optimal coherent-classical estimator is defined as the systems (\ref{eq:coh_controller}), (\ref{eq:coh_class_estimator}) which together minimize the quantity (\ref{eq:est_cost}).

We can now combine the quantum plant (\ref{eq:plant}) and the coherent controller (\ref{eq:coh_controller}) to yield an augmented quantum linear system defined by the following QSDEs:

\small
\begin{equation}\label{eq:coh_system}
\begin{split}
\left[\begin{array}{c}
da\\
da^{\#}\\
da_c\\
da_c^{\#}
\end{array}\right] &= \left[\begin{array}{cc}
F & 0\\
G_cH & F_c
\end{array}\right] \left[\begin{array}{c}
a\\
a^{\#}\\
a_c\\
a_c^{\#}
\end{array}\right] dt + \left[\begin{array}{c}
G\\
G_cK
\end{array}\right] \left[\begin{array}{c}
d\mathcal{A}\\
d\mathcal{A}^{\#}
\end{array}\right];\\
\left[\begin{array}{c}
d\tilde{\mathcal{Y}}\\
d\tilde{\mathcal{Y}}^{\#}
\end{array}\right] &= \left[\begin{array}{cc}
K_cH & H_c
\end{array}\right] \left[\begin{array}{c}
a\\
a^{\#}\\
a_c\\
a_c^{\#}
\end{array}\right] dt + K_cK \left[\begin{array}{c}
d\mathcal{A}\\
d\mathcal{A}^{\#}
\end{array}\right].
\end{split}
\end{equation}
\normalsize

The optimal classical estimator is given by the standard (complex) Kalman filter defined for the system (\ref{eq:coh_system}), (\ref{eq:coh_class_hd}). This optimal classical estimator is obtained from the solution $\tilde{P}_e$ to an algebraic Riccati equation of the form (\ref{eq:riccati}), where
\begin{equation}\label{eq:coh_sys_matrices}
\begin{split}
F_a &= \left[\begin{array}{cc}
F & 0\\
G_cH & F_c
\end{array}\right], \qquad G_a = \left[\begin{array}{c}
G\\
G_cK
\end{array}\right],\\
H_a &= \left[\begin{array}{cc}
K_cH & H_c
\end{array}\right], \qquad K_a = K_cK,\\
L &= \left[\begin{array}{cc}
\tilde{L}_1 & \tilde{L}_2
\end{array}\right],\\
\tilde{L}_1 &= \left[\begin{array}{cccc}
\cos(\tilde{\theta}_1) & 0 & \hdots & 0\\
0 & \cos(\tilde{\theta}_2) & \hdots & 0\\
 & & \ddots & \\
 & & & \cos(\tilde{\theta}_{\tilde{m}})
\end{array}\right],\\
\tilde{L}_2 &= \left[\begin{array}{cccc}
\sin(\tilde{\theta}_1) & 0 & \hdots & 0\\
0 & \sin(\tilde{\theta}_2) & \hdots & 0\\
 & & \ddots & \\
 & & & \sin(\tilde{\theta}_{\tilde{m}})
\end{array}\right],
\end{split}
\end{equation}
where since the quantum noise $\mathcal{A}$ is assumed to be purely canonical, i.e. $d\mathcal{A}d\mathcal{A}^\dagger = Idt$, we have $K_a=K_cK=I$, which requires $K_c=I$ too, as $K=I$.

Then, the corresponding optimal classical estimator (\ref{eq:coh_class_estimator}) is defined by the equations:
\begin{equation}\label{eq:coh_sys_est}
\begin{split}
\tilde{F}_e &= F_a - \tilde{G}_eLH_a;\\
\tilde{G}_e &= (G_aK_a^\dagger + \tilde{P}_eH_a^\dagger)L^\dagger(LK_aK_a^\dagger L^\dagger)^{-1};\\
\tilde{H}_e &= \left[\begin{array}{cc}
C & 0
\end{array}\right].
\end{split}
\end{equation}

We write:
\begin{equation}
\tilde{P}_e = \left[\begin{array}{cc}
P_1 & P_2\\
P_2^\dagger & P_3
\end{array}\right],
\end{equation}
where $P_1$ is of the same dimension as $\bar{P}_e$.

Then, the value of the corresponding cost of the form (\ref{eq:est_cost}) is then given by
\begin{equation}\label{eq:coh_cost}
\tilde{J}_c = \left[\begin{array}{cc}
C & 0
\end{array}\right] \tilde{P}_e \left[\begin{array}{c}
C^\dagger \\
0
\end{array}\right] = CP_1C^\dagger .
\end{equation}

Also, we calculate:

\small
\begin{equation}
\begin{split}
G_aG_a^\dagger &= \left[\begin{array}{cc}
GG^\dagger & GG_c^\dagger\\
G_cG^\dagger & G_cG_c^\dagger
\end{array}\right],\\
G_aK_a^\dagger + \tilde{P}_eH_a^\dagger &= \left[\begin{array}{c}
G + P_1H^\dagger + P_2H_c^\dagger\\
G_c + P_2^\dagger H^\dagger + P_3H_c^\dagger
\end{array}\right],\\
F_a\tilde{P}_e &= \left[\begin{array}{cc}
FP_1 & FP_2\\
G_cHP_1+F_cP_2^\dagger & G_cHP_2+F_cP_3
\end{array}\right],\\
\tilde{P}_eF_a^\dagger &= \left[\begin{array}{cc}
P_1F^\dagger & P_1H^\dagger G_c^\dagger + P_2F_c^\dagger\\
P_2^\dagger F^\dagger & P_2^\dagger H^\dagger G_c^\dagger + P_3F_c^\dagger
\end{array}\right].
\end{split}
\end{equation}
\normalsize

Thus, upon expanding the Riccati equation (\ref{eq:riccati}), we get the following set of equations:

\small
\begin{equation}\label{eq:riccati_expand}
\begin{split}
&FP_1+P_1F^\dagger +GG^\dagger -(G+P_1H^\dagger +P_2H_c^\dagger)L^\dagger\\
&\times L(G + P_1H^\dagger + P_2H_c^\dagger)^\dagger =0,\\
&FP_2+P_1H^\dagger G_c^\dagger +P_2F_c^\dagger +GG_c^\dagger -(G+P_1H^\dagger +P_2H_c^\dagger)L^\dagger \\
&\times L(G_c +P_2^\dagger H^\dagger +P_3H_c^\dagger)^\dagger = 0,\\
&G_cHP_2 + P_2^\dagger H^\dagger G_c^\dagger + F_cP_3 + P_3F_c^\dagger + G_cG_c^\dagger\\
&-(G_c +P_2^\dagger H^\dagger +P_3H_c^\dagger) L^\dagger L(G_c +P_2^\dagger H^\dagger +P_3H_c^\dagger)^\dagger = 0.
\end{split}
\end{equation}
\normalsize

\section{Main Result}
Our main result is the following theorem.

\begin{theorem}\label{thm:central_result}
Consider a coherent-classical estimation scheme defined by (\ref{eq:ann_lqs_1}) ($\mathcal{A}^{out}$ being $\mathcal{Y}$), (\ref{eq:coh_controller}), (\ref{eq:coh_class_hd}) and (\ref{eq:coh_class_estimator}), such that the plant is physically realizable, with the estimation error cost $\tilde{J}_c$ defined in (\ref{eq:coh_cost}). Also, consider the corresponding purely-classical estimation scheme defined by (\ref{eq:ann_lqs_1}), (\ref{eq:class_hd}) and (\ref{eq:class_estimator}), such that the plant is physically realizable, with the estimation error cost $\bar{J}_c$ defined in (\ref{eq:class_cost}). Then,
\begin{equation}
\tilde{J}_c = \bar{J}_c.
\end{equation}
\end{theorem}

\begin{proof}
We first consider the form of the system (\ref{eq:plant}) under the assumption that the plant can be characterized purely by annihilation operators. This essentially implies that the plant is a passive quantum system. A quantum system (\ref{eq:lqs_1}), (\ref{eq:lqs_2}) is characterized by annihilation operators only when $F_2, G_2, H_2, K_2 = 0$.

Then, the equations for the annihilation operators in (\ref{eq:plant}) take the form:
\begin{equation}\label{eq:ann_plant}
\begin{split}
da &= F_1adt + G_1d\mathcal{A};\\
d\mathcal{Y} &= H_1adt + K_1d\mathcal{A}.
\end{split}
\end{equation}

The corresponding equations for the creation operators are then:
\begin{equation}\label{eq:crn_plant}
\begin{split}
da^{*} &= F_1^{\#}a^{*}dt + G_1^{\#}d\mathcal{A}^{\#};\\
d\mathcal{Y}^{\#} &= H_1^{\#}a^{*}dt + K_1^{\#}d\mathcal{A}^{\#},
\end{split}
\end{equation}
Hence, the plant is described by (\ref{eq:plant}) where:
\begin{equation}
\begin{split}
F &= \left[\begin{array}{cc}
F_1 & 0\\
0 & F_1^{\#}
\end{array}\right], G = \left[\begin{array}{cc}
G_1 & 0\\
0 & G_1^{\#}
\end{array}\right],\\
H &= \left[\begin{array}{cc}
H_1 & 0\\
0 & H_1^{\#}
\end{array}\right], K = \left[\begin{array}{cc}
K_1 & 0\\
0 & K_1^{\#}
\end{array}\right].
\end{split}
\end{equation}

Next, we use the assumption that the plant is physically realizable. Then, by applying Theorem \ref{thm:ann_phys_rlz} to (\ref{eq:ann_plant}), there exists a matrix $\Theta_1>0$, such that:
\begin{equation}\label{eq:ann_phys_rlz}
\begin{split}
F_1\Theta_1 +\Theta_1 F_1^\dagger +G_1G_1^\dagger &=0,\\
G_1&=-\Theta_1 H_1^\dagger ,\\
K_1&=I.
\end{split}
\end{equation}

Hence,
\begin{equation}\label{eq:crn_phys_rlz}
\begin{split}
F_1^{\#}\Theta_1^{\#} +\Theta_1^{\#} F_1^T +G_1^{\#}G_1^T &=0,\\
G_1^{\#}&=-\Theta_1^{\#} H_1^T ,\\
K_1^{\#}&=I.
\end{split}
\end{equation}

Combining (\ref{eq:ann_phys_rlz}) and (\ref{eq:crn_phys_rlz}), we get:
\begin{equation}\label{eq:plant_phys_rlz}
\begin{split}
F\Theta +\Theta F^\dagger +GG^\dagger &=0,\\
G&=-\Theta H^\dagger ,\\
K&=I,
\end{split}
\end{equation}
where $\Theta = \left[\begin{array}{cc}
\Theta_1 & 0\\
0 & \Theta_1^{\#}
\end{array}\right] > 0$.

Substituting $\bar{P}_e = \Theta$ in the left-hand side of the Riccati equation (\ref{eq:class_riccati}), we get:
\[ 
F\Theta + \Theta F^\dagger + GG^\dagger - (G + \Theta H^\dagger)L^\dagger L(G + \Theta H^\dagger)^\dagger ,
\]
which is clearly zero, owing to (\ref{eq:plant_phys_rlz}). Thus, $\Theta$ satisfies (\ref{eq:class_riccati}).

Also, it follows from the above that $\tilde{P}_e = \left[\begin{array}{cc}
\Theta & 0\\
0 & P_3 \end{array}\right]$ satisfies (\ref{eq:riccati}) for the case of coherent-classical estimation, since (\ref{eq:riccati_expand}) is satisfied, owing to (\ref{eq:plant_phys_rlz}), by substituting $P_1 = \Theta$ and $P_2 = 0$ to yield:

\small
\begin{equation}
\begin{split}
&F\Theta +\Theta F^\dagger +GG^\dagger -(G+\Theta H^\dagger)L^\dagger L(G + \Theta H^\dagger)^\dagger =0,\\
&(G+\Theta H^\dagger)G_c^\dagger -(G+\Theta H^\dagger)L^\dagger L(G_c +P_3H_c^\dagger)^\dagger = 0,\\
&F_cP_3 + P_3F_c^\dagger + G_cG_c^\dagger -(G_c +P_3H_c^\dagger) L^\dagger L(G_c +P_3H_c^\dagger)^\dagger = 0,
\end{split}
\end{equation}
\normalsize
where $P_3>0$ is simply the error-covariance of the purely-classical estimation of the coherent controller alone.

Thus, we get $\bar{J}_c = \tilde{J}_c = C\Theta C^\dagger $.
\end{proof}

\begin{remark}
We note that the Kalman gain of the purely-classical estimator is zero when $\bar{P}_e = \Theta$. This implies that the Kalman state estimate is independent of the measurement. This is consistent with Corollary 1 of Ref. \cite{IRP3}, which states that for a physically realizable annihilation operator quantum system with only quantum noise inputs, any output field contains no information about the internal variables of the system.
\end{remark}

\begin{remark}
Theorem \ref{thm:central_result} implies that a coherent-classical estimation of a physically realizable quantum plant, that can be described purely by annihilation operators, performs exactly identical to, and no better than, a purely-classical estimation of the plant. This is so because the output field of the quantum plant, as observed above, contains no information about the internal variables of the plant and, therefore, serves simply as a quantum white noise input for the coherent controller.
\end{remark}

\vspace*{2mm}

It is interesting to ask at this point that if the plant is not a passive quantum system and, therefore, cannot be described purely by annihilation operators, is it possible to get a better estimate with a coherent-classical estimator than that with a purely-classical estimator? It turns out that it is possible to get a lower estimation error with coherent-classical estimation (even without the coherent feedback used in \cite{IRP2}) than with classical-only estimation. We will see an example in the following section, where this is the case.

\section{Dynamic Squeezer Example}
In this section, we consider examples involving dynamic squeezers. First, we present an example to illustrate Theorem \ref{thm:central_result}.

Let us consider the quantum plant to be described by the QSDEs:

\small
\begin{equation}\label{eq:sqz_plant}
\begin{split}
\left[\begin{array}{c}
da(t)\\
da(t)^{*}
\end{array}\right] &= \left[\begin{array}{cc}
-\frac{\gamma}{2} & -\chi\\
-\chi^{*} & -\frac{\gamma}{2}
\end{array}\right] \left[\begin{array}{c}
a(t)\\
a(t)^{*}
\end{array}\right] dt - \sqrt{\kappa} \left[\begin{array}{c}
d\mathcal{A}(t)\\
d\mathcal{A}(t)^{\#}
\end{array}\right];\\
\left[\begin{array}{c}
d\mathcal{Y}(t)\\
d\mathcal{Y}(t)^{\#}
\end{array}\right] &= \sqrt{\kappa} \left[\begin{array}{c}
a(t)\\
a(t)^{*}
\end{array}\right] dt + \left[\begin{array}{c}
d\mathcal{A}(t)\\
d\mathcal{A}(t)^{\#}
\end{array}\right];\\
z &= \left[\begin{array}{cc}
0.2 & -0.2
\end{array}\right] \left[\begin{array}{c}
a\\
a^{*}
\end{array}\right].
\end{split}
\end{equation}
\normalsize

Here, we choose $\gamma = 4$, $\kappa = 4$ and $\chi = 0$. Note that this system is physically realizable, since $\gamma = \kappa$, and is annihilation operator only, since $\chi = 0$. In fact, this system corresponds to a passive optical cavity. The matrices corresponding to the system (\ref{eq:plant}) are:
\[F = \left[\begin{array}{cc}
-2 & 0\\
0 & -2
\end{array}\right], G = \left[\begin{array}{cc}
-2 & 0\\
0 & -2
\end{array}\right], H = \left[\begin{array}{cc}
2 & 0\\
0 & 2
\end{array}\right],\] 
\[K = I, C = \left[\begin{array}{cc}
0.2 & -0.2
\end{array}\right].\]

We then calculate the optimal classical-only state estimator and the error $\bar{J}_c$ of (\ref{eq:class_cost}) for this system using the standard Kalman filter equations corresponding to homodyne detector angles varying from $\theta = 0^{\circ}$ to $\theta = 180^{\circ}$.

We next consider the case where a dynamic squeezer is used as the coherent controller in a coherent-classical estimation scheme. In this case, the coherent controller (\ref{eq:coh_controller}) is described by the QSDEs:

\small
\begin{equation}\label{eq:sqz_ctrlr}
\begin{split}
\left[\begin{array}{c}
da(t)\\
da(t)^{*}
\end{array}\right] &= \left[\begin{array}{cc}
-\frac{\gamma}{2} & -\chi\\
-\chi^{*} & -\frac{\gamma}{2}
\end{array}\right] \left[\begin{array}{c}
a(t)\\
a(t)^{*}
\end{array}\right] dt - \sqrt{\kappa} \left[\begin{array}{c}
d\mathcal{Y}(t)\\
d\mathcal{Y}(t)^{\#}
\end{array}\right];\\
\left[\begin{array}{c}
d\tilde{\mathcal{Y}}(t)\\
d\tilde{\mathcal{Y}}(t)^{\#}
\end{array}\right] &= \sqrt{\kappa} \left[\begin{array}{c}
a(t)\\
a(t)^{*}
\end{array}\right] dt + \left[\begin{array}{c}
d\mathcal{Y}(t)\\
d\mathcal{Y}(t)^{\#}
\end{array}\right].
\end{split}
\end{equation}
\normalsize

Here, we choose $\gamma = 16$, $\kappa = 16$ and $\chi = 2$, so that the system is physically realizable. The matrices corresponding to the system (\ref{eq:coh_controller}) are:
\[F_c = \left[\begin{array}{cc}
-8 & -2\\
-2 & -8
\end{array}\right], G_c = \left[\begin{array}{cc}
-4 & 0\\
0 & -4
\end{array}\right],\] \[H_c = \left[\begin{array}{cc}
4 & 0\\
0 & 4
\end{array}\right], K_c = I.\]

Then, the classical estimator for this case is calculated according to (\ref{eq:coh_system}), (\ref{eq:coh_sys_matrices}), (\ref{eq:riccati}), (\ref{eq:coh_sys_est}) for the homodyne detector angle varying from $\theta=0^{\circ}$ to $\theta=180^{\circ}$. The resulting value of the cost $\tilde{J}_c$ in (\ref{eq:coh_cost}) alongwith the cost for the purely-classical estimator case is shown in Fig. \ref{fig:cls_vs_coh_cls_1}.

\begin{figure}[!b]
\hspace*{-5mm}
\includegraphics[width=0.56\textwidth]{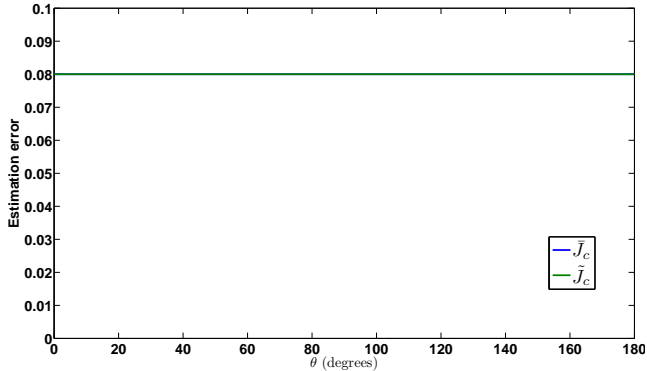}
\caption{Estimation error vs. homodyne detection angle $\theta$ in the case of an annihilation operator only plant.}
\label{fig:cls_vs_coh_cls_1}
\end{figure}

It is clear from the figure that both the classical-only and coherent-classical estimators have the same estimation error cost for all homodyne angles. This illustrates Theorem \ref{thm:central_result} proved in the previous section.

Now, we shall illustrate an example where both the plant and controller are physically realizable, but the plant has $\chi \neq 0$, i.e. it is not annihilation operator only.

In (\ref{eq:sqz_plant}), we choose $\gamma = 4$, $\kappa = 4$ and $\chi = 1$. Note that this system is physically realizable, since $\gamma = \kappa$. The matrices corresponding to the system (\ref{eq:plant}) are:
\[F = \left[\begin{array}{cc}
-2 & -1\\
-1 & -2
\end{array}\right], G = \left[\begin{array}{cc}
-2 & 0\\
0 & -2
\end{array}\right], H = \left[\begin{array}{cc}
2 & 0\\
0 & 2
\end{array}\right],\] 
\[K = I, C = \left[\begin{array}{cc}
0.2 & -0.2
\end{array}\right].\]

We then calculate the optimal classical estimator as above for the homodyne detector angle varying from $\theta=0^{\circ}$ to $\theta=180^{\circ}$.

Next, in (\ref{eq:sqz_ctrlr}), we choose $\gamma = 16$, $\kappa = 16$ and $\chi = 4$. Note that this system is also physically realizable, since $\gamma = \kappa$. The matrices corresponding to the system (\ref{eq:coh_controller}) are:
\[F_c = \left[\begin{array}{cc}
-8 & -4\\
-4 & -8
\end{array}\right], G_c = \left[\begin{array}{cc}
-4 & 0\\
0 & -4
\end{array}\right],\] \[H_c = \left[\begin{array}{cc}
4 & 0\\
0 & 4
\end{array}\right], K_c = I.\]

The classical estimator for this case is calculated as above for the homodyne detector angle varying from $\theta=0^{\circ}$ to $\theta=180^{\circ}$. The resulting value of the cost $\tilde{J}_c$ in (\ref{eq:coh_cost}) alongwith the cost $\bar{J}_c$ in (\ref{eq:class_cost}) for the purely-classical estimator case is shown in Fig. \ref{fig:cls_vs_coh_cls_2}.

\begin{figure}[!b]
\hspace*{-5mm}
\includegraphics[width=0.56\textwidth]{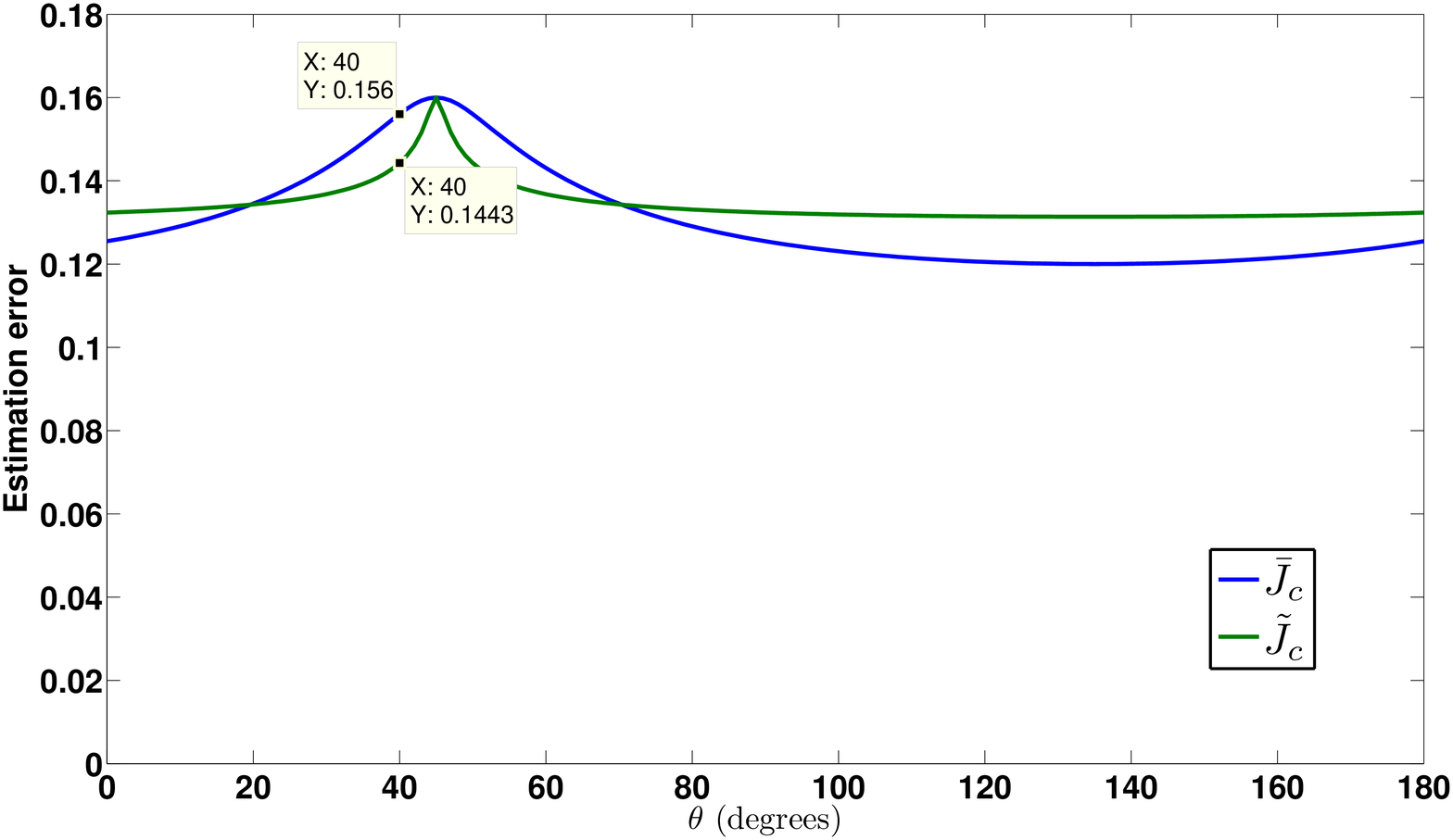}
\caption{Estimation error vs. homodyne detection angle $\theta$ in the case of a squeezer plant.}
\label{fig:cls_vs_coh_cls_2}
\end{figure}

From this figure, we can see that the coherent-classical estimator can perform better than the purely-classical estimator, e.g. for a homodyne detector angle of $\theta=40^{\circ}$. It, however, appears that, for the best choice of homodyne detector angle, the classical-only estimator always outperforms the coherent-classical estimator.

\section{Conclusion}
In this paper, we have demonstrated that a coherent-classical estimation scheme without coherent feedback cannot provide better estimates than a classical-only estimation scheme, when the quantum plant is assumed to be a physically realizable annihilation operator only quantum system. We have also presented an example where such coherent-classical estimation is better than purely-classical estimation, if the plant is not assumed to be an annihilation operator only quantum system.

\bibliographystyle{IEEEtran}
\bibliography{IEEEabrv,cdcbib}

\end{document}